\documentclass[journal, draftclsnofoot,oneside, onecolumn,11pt, letterpaper]{IEEEtran}

\ifCLASSINFOpdf
\else
\fi
\usepackage{pslatex}
\usepackage{amsfonts,color,morefloats}
\usepackage{amssymb,amsmath,latexsym,amsthm}
\usepackage{bbm}
\usepackage{amsfonts}
\usepackage{mathrsfs}
\usepackage{amsthm}
\usepackage{latexsym}
\usepackage{color}
\usepackage{dcolumn}
\usepackage{extarrows}
\usepackage{amsmath}
\allowdisplaybreaks[4]

\newtheorem{theorem}{Theorem}
\newtheorem{lemma}[theorem]{Lemma}

\newtheorem{definition}{Definition}
\newtheorem{example}{Example}
\newtheorem{remark}{Remark}

\begin{document}
%
\title{On the 4-Adic Complexity of Quaternary Sequences with Ideal Autocorrelation}
%
%
%

\author{Minghui~Yang,
        Shiyuan~Qiang,
        Xiaoyan~Jing,
        Keqin~Feng,
        and Dongdai~Lin

\thanks{The material about Theorems 5-6 (excluding examples) in this paper was presented at the IEEE International Symposium on Information Theory (ISIT), Victoria, Australia, July 12-July 20, 2021.}
\thanks{The work was supported by the State Key Program of National Natural Science Foundation of China under Grant 12031011.}
\thanks{Minghui Yang and Dongdai Lin are with State Key Laboratory of Information Security, Institute of Information Engineering, Chinese Academy of Sciences, Beijing 100093, China (e-mail:  yangminghui6688@163.com; ddlin@iie.ac.cn).}
\thanks{Shiyuan Qiang is with the Department of Applied Mathematics, China Agricultural University, Beijing 100083, China (e-mail: qsycau\_18@163.com).}
\thanks{Xiaoyan Jing is with the Research Center for Number Theory and Its Applications, Northwest University, Xi'an 710127, China (e-mail: jxymg@126.com).}
\thanks{Keqin Feng is with the Department of Mathematical Sciences, Tsinghua University, Beijing 100084, China (e-mail: fengkq@tsinghua.edu.cn).}
}

%
%

\markboth{}%
{Shell \MakeLowercase{\textit{et al.}}: Bare Demo of IEEEtran.cls for IEEE Journals}
%



\maketitle

\begin{abstract}
In this paper, we determine the 4-adic complexity of the balanced quaternary sequences of period $2p$ and $2(2^n-1)$ with ideal autocorrelation defined by Kim et al. (ISIT, pp. 282-285, 2009) and Jang et al. (ISIT, pp. 278-281, 2009), respectively. Our results
show that the 4-adic complexity of the quaternary
sequences defined in these two papers is large enough to resist the attack of the rational
approximation algorithm.

\end{abstract}


%
\IEEEpeerreviewmaketitle

\section{Introduction}

With the development of correlation attack and algebraic attack, it is becoming the main trend to use the nonlinear feedback shift register
sequences with pseudorandom property as the driving sequences in stream
cipher design. The feedback with carry shift register (FCSR) proposed by \cite{K4} and \cite{K2} is a kind of
generator which can produce nonlinear sequences quickly.

Balanced binary and quaternary sequences with good autocorrelation play important roles in communication and cryptography systems. The $d$-adic complexity $\Phi_d(s)$ measures the smallest length of FCSR which generates the sequence $s$ over $\mathbf{Z}/(d)$. Sequences over $\mathbf{Z}/(d)$ with low $d$-adic complexity are susceptibly decoded by the rational approximation algorithm, see \cite{K4}, [8-9]. Particularly, a quaternary sequence $s$ can be decoded by the rational approximation algorithm with $6\Phi_4(s)+16$ consecutive bits. Hence, the 4-adic complexity $\Phi_4(s)$ of a safe sequence $s$ with period $N$ should exceed $\frac{N-16}{6}$. There are numerous results about the 2-adic complexity of binary sequences with good autocorrelation, see [2-3], [11-14], for example. However, the 4-adic complexity of quaternary sequences with good autocorrelation has not been studied so fully and there are few quaternary sequences with good autocorrelation whose 4-adic complexity is known, see \cite{Q}. This may pose risk to communication and cryptography system.

In this paper, we determine the 4-adic complexity of the balanced quaternary sequences of even period $2p$ and $2(2^n-1)$ with ideal autocorrelation defined in \cite{Kim} and \cite{Jang}, respectively.  Our results show that the 4-adic complexity of the quaternary sequences with period $2p$ and $2(2^n-1)$ defined in these two papers is larger than $\frac{2p-16}{6}$ and $\frac{2(2^n-1)-16}{6}$ respectively. Hence they are safe enough to resist the attack of the rational approximation algorithm.

\section{Preliminaries}\label{sec2}

In the application of communication and cryptography, balanced sequences with good autocorrelation property are preferred.

For a sequence $g=(g_0, g_1, \ldots, g_{N-1})$ over $\mathbf{Z}/(d)$ with period  $N$, it is said to be balanced if $|A_i-A_j|\leq 1$ for any pair of $i, j$ with $0\leq i\neq j\leq N-1$,where $$A_k=\{t|g_t=k, 0\leq t<N \}, \ k=0, 1,\ldots, d-1.$$

The autocorrelation function of a sequence $s=(s_0, s_1, \ldots, s_{N-1})$ over  $\mathbf{Z}/(d)$ with period $N$ is defined by
$$C_s(\tau)=\sum_{i=0}^{N-1}\zeta_d^{s_i-s_{i+\tau}}, \ \ \ 0\leq\tau< N,$$ where $\zeta_d$ is a complex $d$-th primitive root of unity.

 The maximal out-of-phase autocorrelation magnitude should be as small as possible and the number of the occurrences of the maximal out-of-phase autocorrelation magnitude should be minimized. A sequence with the possible minimum value of the maximal out-of-phase autocorrelation magnitude and the minimum number of occurrences of the maximal out-of-phase autocorrelation magnitude is said to have the ideal autocorrelation property.

 For a binary sequence $s$ with period $N$, it is well known that if
\begin{equation}\label{e111}
C_s(\tau)=-1 \ \ \textrm{for all $0<\tau<N$},
\end{equation} then $s$ is an ideal autocorrelation sequence.

The autocorrelation distribution of a quaternary sequence $s$ of even period $N$ with ideal autocorrelation and balance property is given by
\begin{align*}
C_s(\tau)= \left\{ \begin{array}{ll}
N, & \textrm{1 times},\\
0, & \textrm{$\frac{N}{2}-1$ times},\\
-2, & \textrm{$\frac{N}{2}$ times}.
\end{array} \right.
\end{align*} in \cite{Kim}.

By using the Legendre sequences and the Gray mapping, two classes of balanced quaternary sequences of even period $2p$ with ideal autocorrelation were constructed in \cite{Kim}. Balanced quaternary sequences of period $2(2^n-1)$ with ideal autocorrelation were constructed in \cite{Jang} by using the binary sequences of period $2^n-1$ with ideal autocorrelation and the Gray mapping.

For an odd prime $p$, let $QR$ and $QNR$ be the set of quadratic residues and quadratic non-residues in the set $\mathbf{Z}_p^{\ast}={\mathbf{Z}/(p)}\backslash \{0\}=\{1, 2, \ldots, p-1\}$, respectively. Two classes of Legendre sequences $b$ and $c$ of period $p$ are defined by
\begin{align*}
b_t= \left\{ \begin{array}{ll}
0, & \textrm{for $t=0$}\\
0, & \textrm{for $t\in QR$}\\
1, & \textrm{for $t\in QNR$}
\end{array} \right.
\end{align*}
\begin{align*}
c_t= \left\{ \begin{array}{ll}
1, & \textrm{for $t=0$}\\
0, & \textrm{for $t\in QR$}\\
1, & \textrm{for $t\in QNR$}
\end{array} \right.
\end{align*}
respectively.

The Gray mapping $\phi$ is defined by $$\phi(0,0)=0, ~\phi(0,1)=1,~ \phi(1,1)=2, ~\phi(1,0)=3.$$
According to the definition of the Gray mapping, we can get
\begin{equation}\label{e1}
\phi(a,e)=2a-a(e-1)-(a-1)e
\end{equation} where $0\leq a\leq 1$ and $0\leq e\leq 1$.

The following two classes of quaternary sequences $g^1$ and $g^2$ of even period $2p$ defined by using the Gray mapping and the Legendre sequences were shown to have ideal autocorrelation and balance property in \cite{Kim}.
\begin{definition}\label{d1}(\cite{Kim})
For an odd prime $p$ with $p\equiv 1\pmod 4$, let $s^0$ and $s^1$ be two binary sequences of the same period $2p$ defined by
\begin{align*}
s^0_{t}=\left\{ \begin{array}{ll}
b_t, & \textrm{for $t\equiv 0\bmod 2$}\\
c_t, & \textrm{for $t\equiv 1\bmod 2$}
\end{array} \right.
\end{align*}
\begin{align*}
s^1_{t}=\left\{ \begin{array}{ll}
b_{t}, & \textrm{for $t\equiv 0\bmod 2$}\\
1-c_{t}, & \textrm{for $t\equiv 1\bmod 2$}.
\end{array} \right.
\end{align*}
The quaternary sequence $g^1$ of period $2p$ is defined by $g^1_{t}=\phi(s^0_{t}, s^1_{t})$.
\end{definition}

\begin{definition}\label{d2}(\cite{Kim})
For an odd prime $p$ with $p\equiv 3\pmod 4$, let $s^2$ and $s^3$ be two binary sequences of the same period $2p$ defined by
\begin{align*}
s^2_{t}=\left\{ \begin{array}{ll}
b_{t}, & \textrm{for $0\leq t< p$}\\
b_{t}, & \textrm{for $p\leq t< 2p$}
\end{array} \right.
\end{align*}
\begin{align*}
s^3_{t}=\left\{ \begin{array}{ll}
c_{t}, & \textrm{for $t\equiv 0\bmod 2$}\\
1-c_{t}, & \textrm{for $t\equiv 1\bmod 2$}.
\end{array} \right.
\end{align*}
The quaternary sequence $g^2$ of period $2p$ is defined by $g^2_{t}=\phi(s^2_{t}, s^3_{t})$.
\end{definition}

Let $\mathbf{Z}_{2^n-1}=\mathbf{Z}/(2^n-1)=\{0, 1, 2, \ldots, 2^n-2\}$. Assume that $s$ is a binary sequence of period $2^n-1$ with ideal autocorrelation. Let $D_0$ be the characteristic set of $s$ defined by $$D_0=\{t|s_t=1, 0\leq t\leq 2^n-2\}$$ and $\overline{D}_0=\mathbf{Z}_{2^n-1}\backslash D_0$. By the Chinese remainder theorem, we have the isomorphism
$$\phi: \mathbf{Z}_{2\times(2^n-1)}\simeq \mathbf{Z}_2\times \mathbf{Z}_{2^n-1}, h\mapsto(h\bmod 2, h\bmod {2^n-1}).$$

The following class of quaternary sequences $g^3$ of even period $2(2^n-1)$ defined by using the Gray mapping and the ideal autocorrelation sequences with period $2^n-1$ were shown to have ideal autocorrelation and balance property in \cite{Jang}.

\begin{definition}(\cite{Jang})\label{d3}
Let $s$ be binary sequence of period $2^n-1$ with ideal autocorrelation and $D_0$ a characteristic set of $s$. Let $g^3$ be the quaternary sequence defined by $$g^3_t=\phi(u_t, v_t),$$ where $u$ and $v$ are the binary sequences of period $2^{n+1}-2$ defined by
\begin{align*}
u_{t}=\left\{ \begin{array}{ll}
1, & \textrm{if $t\in\{0,1\}\times D_0$}\\
0, & \textrm{if $t\in\{0,1\}\times \overline{D}_0$}
\end{array} \right.
\end{align*}
\begin{align*}
v_{t}=\left\{ \begin{array}{ll}
1, & \textrm{if $t\in\{0\}\times D_0 \ \bigcup \ \{1\}\times\overline{D}_0$}\\
0, & \textrm{if $t\in\{0\}\times \overline{D}_0\ \bigcup \ \{1\}\times D_0.$}
\end{array} \right.
\end{align*}
\end{definition}
The definition about the $4$-adic complexity of quaternary sequences with period $N$ is defined as follows.
\begin{definition}(\label{d4}\cite{K4,K1})
For a quaternary sequence $s=(s_0, s_1, \ldots, s_{N-1})$ with period $N$, let $S(4)=\sum_{i=0}^{N-1}s_i4^i$.
The 4-adic complexity $\Phi_4(s)$ is defined by $\log_4\frac{4^N-1}{\gcd(4^N-1,\ S(4))},$ where $\gcd(a, b)$ denotes the greatest common divisor of $a $ and $b$. (The exact value of the smallest length of FCSR which generates the quaternary sequence is $\lfloor \log_4{\big((4^N-1)/\gcd(4^N-1, S(4))+1\big)}\rfloor$.
\end{definition}
According to Definition \ref{d4}, determining the $4$-adic complexity of quaternary sequences is equivalent to determining $\gcd(4^N-1, S(4))$.

\section{Main result}

In this section, we study the 4-adic complexity of the quaternary sequences of period $2p$ and $2(2^n-1)$ with ideal autocorrelation in Section \ref{sec2}.

For $i\in \mathbf{Z}_p^{\ast}$, the Legendre symbol $\left(\frac{i}{p}\right)$  is  defined by
\begin{align*}
\left(\frac{i}{p}\right)
 = \left\{ \begin{array}{ll}
1, & \textrm{if $i\in QR$ }\\
-1, & \textrm{otherwise}.
\end{array} \right.
\end{align*}
The following four lemmas are useful in the sequel.
\begin{lemma}(\cite{G}, Theorem 7.3)\label{lem0}
If $s$ is a periodic binary sequence of odd period $2^n-1$ with ideal autocorrelation, then the number of nonzero bits in one period of $s$ is $2^{n-1}$.
\end{lemma}

The proof of the lemma is similar to that of Lemma 2(1) in \cite{LZ}. For the
completeness of the paper, we give a simple proof.

\begin{lemma}\label{lem1}
Let $p$ be an odd prime. Then $$\left(\sum_{i=1}^{p-1}\left(\frac{i}{p}\right)4^{i}\right)^2\equiv-\left(\frac{-1}{p}\right)\frac{4^p-1}{3}+\left(\frac{-1}{p}\right)p\pmod {4^{p}-1}.$$
\end{lemma}
\begin{proof}
Since
\begin{align*}
\left(\sum_{i=1}^{p-1}\left(\frac{i}{p}\right)4^{i}\right)^2
&=\sum_{i=1}^{p-1}\left(\frac{i}{p}\right)4^{i}\sum_{j=1}^{p-1}\left(\frac{j}{p}\right)4^{j}\\
&=\sum_{i,j=1}^{p-1}\left(\frac{ij}{p}\right)4^{i+j}\ \ (\textrm{let}~ j=ik)\\
&=\sum_{i,k=1}^{p-1}\left(\frac{k}{p}\right)4^{i(k+1)}\\
&=\sum_{k=1}^{p-2}\left(\frac{k}{p}\right)\sum_{i=1}^{p-1}4^{i(k+1)}+\sum_{i=1}^{p-1}4^{ip}\left(\frac{p-1}{p}\right).
\end{align*}
Then from $\sum_{k=1}^{p-2}\left(\frac{k}{p}\right)=-\left(\frac{-1}{p}\right)$ and $$\sum_{i=1}^{p-1}4^{i(k+1)}\equiv\sum_{i=1}^{p-1}4^{i}\pmod{4^p-1}~(1\leq k\leq p-2),$$ we get
\begin{align*}
\left(\sum_{i=1}^{p-1}\left(\frac{i}{p}\right)4^{i}\right)^2
&\equiv -\left(\frac{-1}{p}\right)\bigg(\frac{4^p-1}{3}-1\bigg)+(p-1)\left(\frac{-1}{p}\right)\\
&\equiv -\left(\frac{-1}{p}\right)\frac{4^p-1}{3}+\left(\frac{-1}{p}\right)p\pmod{4^p-1}.
\end{align*}
\end{proof}

\begin{lemma}\label{lem3}
For a prime $p$, if $25|(4^p+1)$, then we have $p=5$.
\end{lemma}
\begin{proof}
Since $4^{10}\equiv 1\pmod{25}$ and $4^i\equiv -1\pmod {25}$ has only one solution $i=5$ in the set $\{i|1\leq i\leq 9\}$, then from $4^p\equiv -1\pmod{25}$, we get $p=5+10k\ (k\in \mathbf{Z})$ which implies $5|p$. Hence we get $p=5$.
\end{proof}
\begin{lemma}\label{lem2}
For an odd prime $p$, we have
$$\sum_{t=1}^{p-1}\left(\frac{2t}{p}\right)4^{2t}+\sum_{\substack{t=0\\t\neq\frac{p-1}{2}}}^{p-1}\left(\frac{2t+1}{p}\right)4^{2t+1}
\equiv
 \left\{ \begin{array}{ll}
2\sum\limits_{t=1}^{p-1}\left(\frac{t}{p}\right)4^{t}\pmod{4^p-1}\\
0\pmod{4^p+1}
\end{array} \right.$$
\end{lemma}

\begin{proof}
From
$$\sum_{t=1}^{p-1}\left(\frac{2t}{p}\right)4^{2t}\equiv \sum_{t=1}^{p-1}\left(\frac{t}{p}\right)4^{t} \equiv\sum_{\substack{t=0\\t\neq\frac{p-1}{2}}}^{p-1}\left(\frac{2t+1}{p}\right)4^{2t+1}\pmod{4^p-1}$$
we get
$$\sum_{t=1}^{p-1}\left(\frac{2t}{p}\right)4^{2t}+\sum_{\substack{t=0\\t\neq\frac{p-1}{2}}}^{p-1}\left(\frac{2t+1}{p}\right)4^{2t+1}\equiv
2\sum_{t=1}^{p-1}\left(\frac{t}{p}\right)4^{t}\pmod{4^p-1}.$$
Since
\begin{align*}
\sum_{t=1}^{\frac{p-1}{2}}\left(\frac{2t}{p}\right)4^{2t}\equiv-\sum_{t=\frac{p-3}{2}+2}^{p-1}\left(\frac{2t+1}{p}\right)4^{2t+1}\pmod{4^p+1}
\end{align*}
\begin{align*}
\sum_{t=\frac{p-1}{2}+1}^{p-1}\left(\frac{2t}{p}\right)4^{2t}\equiv-\sum_{t=0}^{\frac{p-3}{2}}\left(\frac{2t+1}{p}\right)4^{2t+1}\pmod{4^p+1}
\end{align*}
then the rest result follows from
\begin{align*}
\sum_{t=1}^{p-1}\left(\frac{2t}{p}\right)4^{2t}&= \sum_{t=1}^{\frac{p-1}{2}}\left(\frac{2t}{p}\right)4^{2t}+ \sum_{t=\frac{p-1}{2}+1}^{p-1}\left(\frac{2t}{p}\right)4^{2t}
\end{align*}
and

 $$\sum_{\substack{t=0\\t\neq\frac{p-1}{2}}}^{p-1}\left(\frac{2t+1}{p}\right)4^{2t+1}=\sum_{t=0}^{\frac{p-3}{2}}\left(\frac{2t+1}{p}\right)4^{2t+1}
+\sum_{t=\frac{p-3}{2}+2}^{p-1}\left(\frac{2t+1}{p}\right)4^{2t+1}. $$

\end{proof}
Now we study the 4-adic complexity of the quaternary sequence $g^1$ in Definition \ref{d1}.
\begin{theorem}\label{t1}
For the quaternary sequence $g^1$ in Definition \ref{d1}, we have
\begin{align*}
\Phi_4(g^1)
 = \left\{ \begin{array}{ll}
\log_4\frac{4^{2p}-1}{15}, & \textrm{if $5|(p+2)$ }\\
\log_4\frac{4^{2p}-1}{3}, & \textrm{else}.
\end{array} \right.
\end{align*}

\end{theorem}

\begin{proof}
(i) Firstly, we prove $$\gcd(g^1(4), 4^{p}-1)=3. $$ Let the symbols be the same as before. Then we get
\begin{align}
g^1(4)&=\sum_{t=0}^{2p-1}\phi(s_t^0, s_t^1)4^t\notag\\& =\sum_{t=0}^{p-1}\phi(b_{2t}, b_{2t})4^{2t}
+\sum_{t=0}^{p-1}\phi(c_{2t+1}, 1-c_{2t+1})4^{2t+1}\notag
\\ &=\sum_{t=0}^{p-1}[2b_{2t}-b_{2t}(b_{2t}-1)-(b_{2t}-1)b_{2t}]4^{2t}+\sum_{t=0}^{p-1}[2c_{2t+1}-c_{2t+1}(-c_{2t+1})\notag
\\ & \ \ -(c_{2t+1}-1)(1-c_{2t+1})]4^{2t+1}\ \ \textrm{(by (\ref{e1}))}\notag
\\ &=\sum_{t=0}^{p-1}[2b_{2t}-{(b_{2t})}^2+b_{2t}-{(b_{2t})}^2+b_{2t}]4^{2t} +\sum_{t=0}^{p-1}[2c_{2t+1}+({c_{2t+1}})^2+1-2c_{2t+1}+({c_{2t+1}})^2]4^{2t+1}\notag
\\ &=\sum_{t=0}^{p-1}2b_{2t}4^{2t}+\sum_{t=0}^{p-1}(2c_{2t+1}+1)4^{2t+1}\ \ \textrm{(since $a^2=a ~(0\leq a\leq 1)$) }\notag
\\ &=2b_{0}4^{0}+\sum_{t=1}^{p-1}2b_{2t}4^{2t}+(2c_p+1)4^p+\sum_{\substack{t=0\\t\neq\frac{p-1}{2}}}^{p-1}(2c_{2t+1}+1)4^{2t+1}\notag
\\ &=\sum_{t=1}^{p-1}2b_{2t}4^{2t}+3\cdot 4^p+\sum_{\substack{t=0\\t\neq\frac{p-1}{2}}}^{p-1}(2c_{2t+1}+1)4^{2t+1}\ \ \textrm{(since $b_0=0$, $c_0=c_p=1$) }\notag
\\ &=\sum_{\substack{t=0\\t\neq\frac{p-1}{2}}}^{p-1}4^{2t+1}+3\cdot 4^p+\sum_{t=1}^{p-1}2b_{2t}4^{2t}+\sum_{\substack{t=0\\t\neq\frac{p-1}{2}}}^{p-1}2c_{2t+1}4^{2t+1}\notag
\\ &= 4\sum_{t=0}^{p-1}4^{2t}-4^p+3\cdot4^p
  +2\sum_{t=1}^{p-1}\frac{1-\left(\frac{2t}{p}\right)}{2}4^{2t} +2\sum_{\substack{t=0\\t\neq\frac{p-1}{2}}}^{p-1}\frac{1-\left(\frac{2t+1}{p}\right)}{2}4^{2t+1}\notag\\
  &\equiv\left\{ \begin{array}{ll}
9\sum\limits_{t=0}^{p-1}4^{2t}-2\sum\limits_{t=1}^{p-1}\left(\frac{t}{p}\right)4^{t}\pmod{4^p-1}\\
9\sum\limits_{t=0}^{p-1}4^{2t}-2\pmod{4^p+1}.
\end{array} \right.\ \ \  \ \textrm{(by Lemma \ref{lem2})}
\end{align}\label{e3}

 Since $3|(4^p-1$), then from (3) we know
\begin{align*}
 g^1(4)\equiv\sum_{t=1}^{p-1}\left(\frac{t}{p}\right)4^t\equiv 0 \pmod 3\ \textrm{\bigg(since $\sum_{t=1}^{p-1}\left(\frac{t}{p}\right)=0$\bigg).}
\end{align*}
It then follows that $3|\gcd(g^1(4), 4^p-1)$. If $9|(4^p-1)$, then from $4^3\equiv 1\pmod 9$ and $4^p\equiv 1\pmod 9$, we get
$p=3$ which contradicts with $p\equiv 1\pmod 4$. Therefore $9\nmid \gcd(g^1(4), 4^p-1)$.

Assume that  $d_1$ is a prime divisor of $\gcd(g^1(4), 4^p-1)$ such that $d_1\neq 3$.  By (3) we get
\begin{align*}
 g^1(4)\equiv-2\sum_{t=1}^{p-1}\left(\frac{t}{p}\right)4^t\pmod{d_1}.
\end{align*}
Then we have $d_1\big|\left(\sum\limits_{t=1}^{p-1}\left(\frac{t}{p}\right)4^t\right)^2$. Combining with Lemma \ref{lem1}, we have $d_1\mid p$ which implies $d_1=p$. Hence, we have $4^p\equiv 1\pmod p$. By Fermat's little Theorem, we get $4^{p-1}\equiv 1\pmod p$. Then we have $p|(p-1)$ which is a contradiction. Hence, we know $d_1=1$.

Therefore we get
\begin{align}
\gcd(g^1(4), 4^{p}-1)=3. \label{e4}
\end{align}
(ii) Next, we prove
\begin{align*}
\gcd(g^1(4), 4^{p}+1)
 = \left\{ \begin{array}{ll}
5, & \textrm{if $5|(p+2)$ }\\
1, & \textrm{else}.
\end{array} \right.
\end{align*}
By (3) we have \begin{align*}
 g^1(4)\equiv-p-2\pmod 5.
\end{align*}
Then we get $5|\gcd(g^1(4), 4^p+1)$ only when $5|(p+2)$.

Assume that $5|(p+2)$ and $25|(4^p+1)$, then by Lemma \ref{lem3} we get $p=5$ which contradicts with $5|(p+2)$. It then follows that $25\nmid\gcd(g^1(4), 4^p+1)$.

Let $d_2$ be a divisor of $\gcd(g^1(4), 4^p+1)$ such that $5\nmid d_2$. Then from (3) we have
$g^1(4)\equiv -2\pmod{d_2}.$
Thus $d_2\mid2$ which implies $d_2=1.$
Therefore
\begin{align}
\gcd(g^1(4), 4^{p}+1)
 = \left\{ \begin{array}{ll}
5, & \textrm{if $5|(p+2)$ }\\
1, & \textrm{else}.
\end{array} \right.
\end{align}
Combining with (4-5) and the definition of the 4-adic complexity, the result is proven.
\end{proof}

The 4-adic complexity of the sequence $g^2$ in Definition 2 is given by the following theorem.

\begin{theorem}\label{t2}
For the quaternary sequence $g^2$ in Definition 2, we have
\begin{align*}
\Phi_4(g^2)
 = \left\{ \begin{array}{ll}
\log_4\frac{4^{2p}-1}{5}, & \textrm{if $5|(p-2)$ }\\
\log_4(4^{2p}-1), & \textrm{else}.
\end{array} \right.
\end{align*}
\end{theorem}
\begin{proof}
(i) Firstly, we determine the exact value of
$\gcd(g^2(4), 4^p-1)$.

 Assume that the symbols are the same as before. Then we have
\begin{align}
g^2(4)\notag &=\sum_{t=0}^{2p-1}\phi(s_t^2, s_t^3)4^t\notag\\& =\sum_{t=0}^{p-1}\phi(b_{2t}, c_{2t})4^{2t}+\sum_{t=0}^{p-1}\phi(b_{2t+1}, 1-c_{2t+1})4^{2t+1}\notag
\\ &=\sum_{t=0}^{p-1}[2b_{2t}-b_{2t}(c_{2t}-1)-(b_{2t}-1)c_{2t}]4^{2t}+\sum_{t=0}^{p-1}[2b_{2t+1}+b_{2t+1}c_{2t+1}\notag
\\ & \ \ -(b_{2t+1}-1)(1-c_{2t+1})]4^{2t+1}\ \ \textrm{(by (\ref{e1}))}\notag
\\ &=(2\times 0-0\times(1-1)-(0-1)\times1)\times 4^0+\sum_{t=1}^{p-1}[2b_{2t}-b_{2t}c_{2t}+b_{2t}-b_{2t}c_{2t}+c_{2t}]4^{2t}\notag
\\ & \ \ +[2b_p+b_pc_p-(b_p-1)(1-c_p)]4^p+\sum_{\substack{t=0\\t\neq\frac{p-1}{2}}}^{p-1}[2b_{2t+1}+b_{2t+1}c_{2t+1}+{(b_{2t+1})}^2+1-2b_{2t+1}]4^{2t+1}\notag
\\ &=1+\sum_{t=1}^{p-1}2b_{2t}4^{2t}+\sum_{\substack{t=0\\t\neq\frac{p-1}{2}}}^{p-1}(2b_{2t+1}+1)4^{2t+1} \ \ \textrm{(since $a^2=a ~(0\leq a\leq 1)$) }\notag
\\&=1+\sum_{t=1}^{p-1}2b_{2t}4^{2t}+4\sum_{t=0}^{p-1}4^{2t}-4^{2\cdot\frac{p-1}{2}+1}+\sum_{\substack{t=0\\t\neq\frac{p-1}{2}}}^{p-1}2b_{2t+1}4^{2t+1}\notag
\\ &= 2\sum_{t=1}^{p-1}\frac{1-\left(\frac{2t}{p}\right)}{2}4^{2t}+2\sum_{\substack{t=0\\t\neq\frac{p-1}{2}}}^{p-1}\frac{1-\left(\frac{2t+1}{p}\right)}{2}4^{2t+1}
 +1-4^p+4\sum_{t=0}^{p-1}4^{2t}\label{e7} \notag\\
&\equiv \left\{ \begin{array}{ll}
9\sum\limits_{t=0}^{p-1}4^{2t}-
2\sum\limits_{t=1}^{p-1}\left(\frac{t}{p}\right)4^t-2\pmod{4^p-1}\\
9\sum\limits_{t=0}^{p-1}4^{2t}+2\pmod{4^p+1}.
\end{array} \right.\textrm{(by Lemma \ref{lem2})}
\end{align}

Since $3|(4^p-1)$, then by (6) we know
$$ g^2(4)\equiv 1\pmod 3.$$
Then we get $3\nmid\gcd(g^2(4), 4^p-1)$.

Let $d_3$ be a prime divisor of $\gcd(g^2(4), 4^p-1)$. From (6) and $d_3\neq 3$ we have
$$g^2(4)\equiv -2(\sum_{t=1}^{p-1}\left(\frac{t}{p}\right)4^t+1)\pmod{d_3}.$$ Then by Lemma \ref{lem1} and $p\equiv 3\pmod 4$, we have $$1\equiv (\sum_{t=1}^{p-1}\left(\frac{t}{p}\right)4^t)^2\equiv -p+\frac{4^p-1}{3}\equiv -p\pmod{d_3}.$$ Hence, we have $d_3|(p+1)$. From $4^{d_3-1}\equiv 1\pmod{d_3}$ and $4^p\equiv 1\pmod{d_3}$ we get $p|(d_3-1)$ which is a contradiction. Therefore
\begin{equation}
\gcd(g^2(4), 4^p-1)=1.
\end{equation}
(ii) Now we determine
$\gcd(g^2(4), 4^p+1)$.

Since $5\mid (4^p+1)$, then by (6) we get
$$g^2(4)\equiv -\sum_{t=0}^{p-1}4^{2t}+2\equiv -p+2\pmod 5.$$
Hence we get $5|\gcd(g^2(4), 4^p+1)$ only when $5|(p-2)$.

Assume that $5|(p-2)$ and $25|(4^p+1)$, then from Lemma \ref{lem3} we have $p=5$ which contradicts with $5|(p-2)$. It then follows that $25\nmid\gcd(g^2(4), 4^p+1)$.
Assume that $d_4$ is a prime divisor of $g^2(4)$ and $4^p+1$ such that $d_4\neq 5$.
Then by (6) we get
\begin{align*}
 g^2(4)\equiv 2\pmod{d_4}
                               \end{align*}
which implies $d_4=1$. Hence we have
\begin{align}
\gcd(g^2(4), 4^p+1)
 = \left\{ \begin{array}{ll}
5, & \textrm{if $5|(p-2)$ }\\
1, & \textrm{otherwise}.
\end{array} \right.\label{e10}
\end{align}
Combining with (7-8) and the definition of the 4-adic complexity, the result is proven.
\end{proof}
The 4-adic complexity of the sequence $g^3$ with period $2^{n+1}-2$ in Definition \ref{d3} is given as follows.
\begin{theorem}\label{t3}
For the quaternary sequence $g^3$ with period $2^{n+1}-2$ in Definition \ref{d3}, we have
$$\Phi(g^3)=\log_45\cdot (4^{2^n-1}-1).$$
\begin{proof}
With the symbols the same as before, we have
\begin{align}
g^3(4) &=\sum_{t=0}^{2^{n+1}-3}\phi(u_t, v_t)4^t\notag\\
& =\sum_{\substack{t=0\\t\in \{1\}\times \overline{D}_0}}^{2^{n+1}-3}4^t+\sum_{\substack{t=0\\t\in \{0\}\times D_0}}^{2^{n+1}-3}2\cdot4^t+\sum_{\substack{t=0\\t\in \{1\}\times D_0}}^{2^{n+1}-3}3\cdot4^t\notag\\
& =\sum_{\substack{t=0\\2\nmid t, s_t=0}}^{2^{n+1}-3}4^t+2\cdot\sum_{\substack{t=0\\2|t, s_t=1}}^{2^{n+1}-3}4^t+3\cdot\sum_{\substack{t=0\\2\nmid t, s_t=1}}^{2^{n+1}-3}4^t\notag\\
& =\sum_{\substack{t=0\\2\nmid t}}^{2^{n+1}-3}(1-s_t)4^t+2\cdot\sum_{\substack{t=0\\2|t}}^{2^{n+1}-3}s_t\cdot4^t+3\cdot\sum_{\substack{t=0\\2\nmid t}}^{2^{n+1}-3}s_t\cdot4^t\notag\\
& =\sum_{\substack{t=0\\2\nmid t}}^{2^{n+1}-3}4^t+2\sum_{\substack{t=0}}^{2^{n+1}-3}s_t\cdot4^t\notag\\
& =\sum_{t=0}^{2^n-2}4^{2t+1}+2\sum_{t=0}^{2^n-2}s_t\cdot4^t\cdot(1+4^{2^n-1})\label{e21}\\
& =4\cdot\frac{4^{2^n-1}+1}{5}\cdot\frac{4^{2^n-1}-1}{3}+2\sum_{t=0}^{2^n-2}s_t\cdot4^t\cdot(1+4^{2^n-1}).\label{e11}
\end{align}
(i) Firstly, we determine
$\gcd(g^3(4), 4^{2^n-1}+1)$.

By (\ref{e11}) and $\frac{4^{2^n-1}-1}{3}=\sum\limits_{t=0}^{2^n-2}4^t\equiv 1\pmod 5$, we have $$g^3(4)\equiv 4\cdot \frac{4^{2^n-1}+1}{5}\pmod {4^{2^{n-1}}+1}.$$
It then follows that
\begin{equation}\label{e19}
\gcd(g^3(4), 4^{2^n-1}+1)=\gcd(4\cdot \frac{4^{2^n-1}+1}{5}, 4^{2^{n-1}}+1)=\frac{4^{2^n-1}+1}{5}.
\end{equation}
(ii) Secondly, we determine
$\gcd(g^3(4), 4^{2^n-1}-1)$.
By (\ref{e11}) and $\frac{4^{2^n-1}+1}{5}=\sum\limits_{t=0}^{2^n-2}(-4)^t \equiv 1\pmod 3$, we have
\begin{equation}\label{e12}
g^3(4)\equiv 4\cdot\sum_{t=0}^{2^n-2}s_t4^t +\frac{4^{2^n-1}-1}{3}\pmod {4^{2^{n-1}}-1}.
\end{equation}
Since $s_t\in \{0, 1\}$, we have $s_t=\frac{1-(-1)^{s_t}}{2}$.  Then from (\ref{e12}), we get
\begin{align*}
g^3(4)&\equiv 4\cdot \sum_{t=0}^{2^n-2}\frac{1-(-1)^{s_t}}{2}\cdot 4^t+\frac{4^{2^n-1}-1}{3}\\
&\equiv 2\cdot\frac{4^{2^n-1}-1}{3}-2\sum_{t=0}^{2^n-2}(-1)^{s_t}\cdot 4^t+\frac{4^{2^n-1}-1}{3}\\
&\equiv -2\sum_{t=0}^{2^n-2}(-1)^{s_t}\cdot 4^t\pmod {4^{2^{n-1}}-1}.
\end{align*}
Hence, we have
\begin{equation}\label{e13}
\gcd(g^3(4),4^{2^{n-1}}-1)=\gcd(-\sum_{t=0}^{2^n-2}(-1)^{s_t}4^t, 4^{2^n-1}-1).
\end{equation}
Let $r$ be a prime divisor of $\gcd(g^3(4),4^{2^{n}-1}-1)$. Then from (\ref{e13}), we get
$$0\equiv g^3(4)\equiv \sum_{t=0}^{2^n-2}(-1)^{s_t}\cdot 4^t\pmod r.$$
It then follows that
\begin{align}
0 &\equiv \sum_{t=0}^{2^n-2}(-1)^{s_t}\cdot 4^t\sum_{l=0}^{2^n-2}(-1)^{s_l}\cdot 4^{-l}\notag\\
& \equiv\sum_{t, l=0}^{2^n-2}(-1)^{s_t+s_l}4^{t-l}\notag\\
& \equiv\sum_{l=0}^{2^n-2}\sum_{f=0}^{2^n-2}(-1)^{s_{f+l}+s_l}4^f\notag\\
& \equiv\sum_{f=0}^{2^n-2}4^f\sum_{l=0}^{2^n-2}(-1)^{s_{f+l}+s_l}\pmod r\notag
\end{align}
Then from the fact that $s$ is a binary sequence of period $2^n-1$ with ideal autocorrelation and (\ref{e111}), we get
\begin{equation}\label{e116}
0\equiv 2^n-1+\sum_{f=1}^{2^n-2}(-1)4^f\equiv 2^n-\frac{4^{2^n-1}-1}{3}\pmod r.
\end{equation}
By (\ref{e21}), we know $g^3(4)\equiv\sum_{t=0}^{2^n-2}s_t+2^n-1\pmod 3$. Since $s$ is a binary sequence of period $2^n-1$ with ideal autocorrelation, then from Lemma \ref{lem0} we have  $\sum_{t=0}^{2^n-2}s_t=2^{n-1}$. It then follows that $g^3(4)\equiv2^{n-1}+2^n-1\equiv 2\pmod 3$.
Therefore we get $3\nmid g^3(4)$. It then follows that $r\neq 3$ which implies $$r\bigg|\frac{4^{2^n-1}-1}{3}.$$ By (\ref{e116}) we get
$0\equiv 2^n\pmod r$ which is a contradiction. Hence, we have
\begin{equation}\label{e18}
\gcd(g^3(4),4^{2^{n-1}}-1)=1.
\end{equation}
Combining with (\ref{e19}), (\ref{e18}) and the definition of the 4-adic complexity, the result is proven.
\end{proof}
\end{theorem}
We give several examples to demonstrates our main results.
\begin{example}
For $p=5\equiv1\pmod 4$, we have $\mathbb{F}_5^{\ast}=\langle2\rangle$ and $b_0=b_5=0, c_0=c_5=1, b_1=b_6=0, c_1=c_6=0, b_2=b_7=1, c_2=c_7=1, b_3=b_8=1, c_3=c_8=1, b_4=b_9=0, c_4=c_9=0$. Then according to the definition of the sequence $g^1$, we get $g^1=(0, 1, 2, 3, 0, 3, 0, 3, 2, 1)$. Then we have
\begin{align*}
&\gcd(g^1(4), 4^{10}-1)\\
=&\gcd(1\times 4+2\times 4^2+3\times4^3+3\times4^5+3\times4^7+2\times4^8+4^9, 4^{10}-1)\\=&3
\end{align*}
which implies $\Phi_4(g^1)=\log_4\frac{4^{10}-1}{3}$. This result is consistent with Theorem \ref{t1}.
\end{example}
\begin{example}
For $p=13\equiv1\pmod 4$, we have $\mathbb{F}_{13}^{\ast}=\langle2\rangle$. According to the definition of the sequence $g^1$, we have $g^1=(0, 1, 2, 1, 0, 3, 2, 3, 2, 1, 0, 3, 0, 3, 0, 3, 0, 1, 2, 3, 2, 3, 0, 1, 2, 1)$. Then we have
\begin{align*}
&\gcd(g^1(4), 4^{26}-1)\\
=&\gcd(1\times 4+2\times 4^2+4^3+3\times4^5+2\times4^6+3\times4^7+2\times4^8+4^9+3\times4^{11}+3\times4^{13}\\
&+3\times4^{15}+4^{17}+2\times4^{18}+3\times4^{19}+2\times4^{20}+3\times4^{21}+4^{23}+2\times4^{24}+4^{25}, 4^{26}-1)\\=&15
\end{align*}
which implies $\Phi_4(g^1)=\log_4\frac{4^{26}-1}{15}$. This result is consistent with Theorem \ref{t1}.
\end{example}
\begin{example}
For $p=3\equiv3\pmod 4$, we have $\mathbb{F}_3^{\ast}=\langle2\rangle$ and $b_0=b_3=0, c_0=c_3=1, b_1=b_4=0, c_1=c_4=0, b_2=b_5=1, c_2=c_5=1$. According to the definition of the sequence $g^2$, we have $g^2=(1, 1, 2, 0, 0, 3)$. Then we get
\begin{align*}
&\gcd(g^2(4), 4^{6}-1)\\
=&\gcd(1+1\times 4+2\times 4^2+3\times4^5, 4^{6}-1)\\=&1
\end{align*}
which implies $\Phi_4(g^2)=\log_4(4^{6}-1)$. This result is consistent with Theorem \ref{t2}.
\end{example}
\begin{example}
For $p=7\equiv 3\pmod 4$, according to the definition of the sequence $g^2$, we have $g^2=(1, 1, 0, 3, 0, 3, 2, 0, 0, 1, 2, 1, 2, 3)$. Then we get
\begin{align*}
&\gcd(g^2(4), 4^{14}-1)\\
=&\gcd(1+1\times 4+3\times 4^3+3\times4^5+2\times4^6+4^9+2\times4^{10}+4^{11}+2\times4^{12}+3\times4^{13}, 4^{14}-1)\\=&5
\end{align*}
which implies $\Phi_4(g^2)=\log_4(\frac{4^{14}-1}{5})$. This result is consistent with Theorem \ref{t2}.

\end{example}

\begin{example}
For $n=4$, we have the binary $m$-sequence $s=(0, 0, 0, 1, 0, 0, 1, 1, 0, 1,  0, 1, 1, 1, 1)$ of period 15,\ according to the definition of the sequence $g^3$, we have $$g^3=(0, 1, 0, 3, 0, 1, 2, 3, 0, 3, 0, 3, 2, 3, 2, 1, 0, 1, 2, 1, 0, 3, 2, 1, 2, 1, 2, 3, 2, 3).$$
 Then we get
\begin{align*}
&\gcd(g^3(4), 4^{30}-1)\\
=&\gcd(4+3\times 4^3+4^5+2\times4^6+3\times4^7+3\times4^9+3\times4^{11}+2\times4^{12}+3\times4^{13}
+2\times4^{14}+4^{15}+4^{17}\\&+2\times4^{18}+4^{19}+3\times4^{21}+2\times4^{22}+4^{23}+2\times4^{24}+4^{25}+2\times4^{26}
+3\times4^{27}+2\times4^{28}+3\times4^{29}, 4^{30}-1)\\=&214748365\\=&\frac{4^{15}+1}{5}
\end{align*}
which implies $\Phi_4(g^3)=\log_4(5\times(4^{15}-1))$. This result is consistent with Theorem \ref{t3}.
\end{example}
\begin{remark}
For a sequence $s$ with period $N$, the 4-adic complexity $\Phi_4(s)$ should exceed $\frac{N-16}{6}$ to resist the rational approximation algorithm.  Theorems 5-7  show that the 4-adic complexity of the balanced quaternary sequences of period $2p$ and $2(2^n-1)$ with ideal autocorrelation defined in \cite{Kim} and \cite{Jang} is larger than $\frac{2p-16}{6}$ and $\frac{2(2^n-1)-16}{6}$ respectively.
\end{remark}

\section{Conclusion}
In this paper, we study the 4-adic complexity of the balanced quaternary sequences
with ideal autocorrelation constructed in \cite{Kim} and \cite{Jang}, respectively. It
turns out that the balanced quaternary sequences
with ideal autocorrelation constructed in these two papers are safe enough to resist the attack of the rational approximation algorithm. It would be interesting to investigate
the 4-adic complexity of more quaternary sequences with
good autocorrelation and balance property.

\end{document}